\documentclass[12pt]{article}

\usepackage[margin=1in]{geometry}
\usepackage{amsmath,amssymb,amsthm}

\usepackage[colorlinks=true,linkcolor=blue,citecolor=blue,urlcolor=blue]{hyperref}
\usepackage{float}
\usepackage{authblk}
\usepackage{graphicx}

\theoremstyle{plain}
\newtheorem{theorem}{Theorem}

\newtheorem{lemma}[theorem]{Lemma}
\newtheorem{proposition}[theorem]{Proposition}

\theoremstyle{definition}

\theoremstyle{remark}
\newtheorem{remark}[theorem]{Remark}

\newcommand{\N}{\mathbb{N}}
\newcommand{\Z}{\mathbb{Z}}
\newcommand{\Q}{\mathbb{Q}}
\numberwithin{equation}{section}
\newcommand{\weight}{\mathrm{wt}}
\newcommand{\kernel}{\mathcal{K}}
\newcommand{\pref}{\mathrm{pref}}
\newcommand{\suff}{\mathrm{suff}}
\def\suchthat{\, : \,}

\providecommand{\keywords}[1]{\smallskip\noindent\textbf{\emph{Keywords:}} #1}

\providecommand{\msc}[1]{\smallskip\noindent\textbf{\emph{2020 Mathematics Subject Classification:}} #1}

\newenvironment{dedication}
 {\thispagestyle{empty}
 \vspace*{.5truecm}
 \itshape 
 \raggedleft 
 }
 {\par 
 \vspace{1truecm} 
 }

\title{An example of an automatic sequence with non-regular abelian complexity}

\author[1]{Jean-Paul Allouche}
\author[2]{Narad Rampersad}
\author[3]{Michel Rigo}
\author[4]{Jeffrey Shallit}
\author[3]{Manon Stipulanti}
\affil[1]{CNRS, IMJ-PRG, Sorbonne, 4 Place Jussieu, 75005 Paris, France\\ \href{mailto:jean-paul.allouche@imj-prg.fr}{\tt jean-paul.allouche@imj-prg.fr}}

\affil[2]{Department of Mathematics and Statistics, University of Winnipeg, 515 Portage Ave., Winnipeg,  MB R3B 2E9, Canada\\ \href{mailto:n.rampersad@uwinnipeg.ca}{\tt n.rampersad@uwinnipeg.ca}}

\affil[3]{Department of Mathematics, University of Liège, Allée de la Découverte 12, 4000 Liège, Belgium\\ \href{mailto:m.rigo@uliege.be}{\tt m.rigo@uliege.be},
\href{mailto:m.stipulanti@uliege.be}{\tt m.stipulanti@uliege.be}}

\affil[4]{School of Computer Science, University of Waterloo, Waterloo, ON N2L 3G1, Canada\\ \href{mailto:shallit@uwaterloo.ca}{\tt shallit@uwaterloo.ca}}

\date{\today}
\begin{document}
\maketitle

\begin{dedication}
In honor of Maxime Crochemore's 80th birthday 
\end{dedication}

\begin{abstract}
Answering a question of Madill and Rampersad, we give an example ada
of a $2$-automatic sequence whose abelian complexity is
not $2$-regular.
\end{abstract}


\msc{Primary 68R15; Secondary 11B85, 11A63.}

\section{Introduction}

{\it Combinatorics on words} studies finite or infinite sequences 
taking their values in a given (usually finite) set, called an alphabet. One branch
deals with finite sequences where the elements can be unordered,
leading to ``abelian concepts''. Two early occurrences of such
notions are \cite[p.~240]{Erdos} and \cite{Parikh}. For 
more about this subject and possible applications, see \cite{CC,Fici&Puzynina:2025}.

\medskip

Here we will address a question about ``abelian complexity''.
Let ${\bf x} = (x(n))_{n \geq 0}$ be a sequence over a finite 
alphabet. The {\it factor complexity} $\rho_{\bf x} (n)$ counts 
the number of distinct length-$n$ blocks appearing in $\bf x$, 
while the {\it abelian complexity} $\rho_{\bf x}^{\rm ab} (n)$ 
counts the number of length-$n$ blocks where the order of the 
symbols does not matter. Both measures have been studied 
extensively in the literature; for abelian complexity, see, for 
example, \cite{Richomme&Saari&Zamboni:2011}.

\medskip

A sequence ${\bf x} = (x(n))_{n \geq 0}$ over a finite alphabet 
is said to be $k$-automatic if there is a finite deterministic automaton 
with output (DFAO) that, on input $n$ expressed in base $k\geq 2$, 
terminates in a state with output $x(n)$. Equivalently (\cite{Cobham:1972} and~\cite[Thm.~6.3.2]{Allouche&Shallit:2003}), a sequence is $k$-automatic
if and only if it is the image, under a coding, of a fixed point of a $k$-uniform morphism (i.e., the image of each letter has length equal to $k$). The automatic sequences form a large and interesting 
class that has been studied for more than 50 years
\cite{Cobham:1972, Allouche&Shallit:2003}.

\medskip

The $k$-regular sequences are a generalization of $k$-automatic 
sequences~\cite{AlloucheShallit,Allouche&Shallit:2003}.  We say that a sequence
${\bf x} = (x(n))_{n \geq 0}$ over $\Z$ is 
$k$-regular if its $k$-kernel
$$
\kernel_k ({\bf x}) := \{ (x(k^i n + e))_{n \geq 0} \suchthat
i \geq 0, \ 0 \leq e < k^i \}
$$
spans a finite-dimensional vector space over $\Q$.  

\medskip

It is known that the factor complexity of a $k$-automatic 
sequence is $k$-regular \cite[Prop.~4.3]{Carpi-D-Alonzo}
(a weaker result had been proven earlier by Moss\'e \cite{Mosse}).
A question in Madill and Rampersad \cite{MadillRampersad} (to which the answer 
was widely conjectured to be yes) asked whether the same 
is true for the abelian complexity of a $k$-automatic 
sequence.  
This holds for several sequences that 
have been investigated, for example, successively:
\begin{itemize}
    \item the Thue--Morse sequence, whose abelian complexity is periodic, and hence $k$-regular for every $k$ 
    \cite[Thm.~3.1]{Richomme&Saari&Zamboni:2011}. That
    early overview also contains two nice results, namely
    \cite[Thm.~1.7 and Thm.~3.3]{Richomme&Saari&Zamboni:2011}. We leave to the reader the
    pleasure of applying those two theorems to construct
    $k$-automatic sequences whose abelian complexity is
    periodic, and hence $k$-regular for every $k$.
    \item the regular paperfolding sequence \cite{MadillRampersad}.
    \item fixed points of $k$-uniform morphisms 
    $\varphi$ on $\{0, 1\}$ such that 
    $\big|\lvert\varphi(0)\rvert_0 - \lvert\varphi(1)\rvert_0\big|\neq 1$, which is a consequence of \cite[Cor.~1]{BSSW}.
    \item the Golay--Shapiro sequence (aka the Rudin--Shapiro sequence) \cite{LCWW}.
    \item the Thue--Morse sequence over a ternary alphabet, which is an easy consequence of \cite[Thm.~4.1]{KaboreKentiega}.
    \item the Cantor sequence \cite[Thm.~2]{CLW}.
    \item the generalized Thue--Morse sequences, which is an easy consequence of \cite[Cor.~1]{ChenWen}.
    \item automatic sequences satisfying certain extra conditions
    \cite[Thm.~2]{ChenWen}.
    \item automatic sequences $\mathbf{s}$ such that
    the Parikh vectors of the length-$n$ prefixes of $\mathbf{s}$ 
    form a synchronized sequence and such that the abelian
    complexity of $\mathbf{s}$ is bounded above by a constant \cite{Shallit2021}.
    \item automatic sequences that are fixed points of Parikh-collinear morphisms \cite[Thm.~3]{Rigo&Stipulanti&Whiteland:2024}.
    \item the $2$-block coding of generalized Cantor 
    sequences $\mathbf{c}$ that have mirror-invariant sets of 
    factors \cite[Thm.~1]{LCW}.
    \end{itemize}
An even stronger conjecture appears in \cite{PRRV}: 
namely, that {\it the $\ell$-abelian complexity of a $k$-automatic 
sequence is $k$-regular.} Several papers give positive answers to this question; see~\cite{CLW,Couvreur2025,Greinecker,LCW,PRRV}.

However, in this note we give an example of a $k$-automatic sequence whose abelian complexity is not $k$-regular (see Theorem~\ref{thm:main}).  This result also refutes a conjecture of Popoli and two of the present authors \cite[Conj.~4]{Popoli&Shallit&Stipulanti:2024} about additive complexity, since for binary sequences additive complexity and abelian complexity coincide.

\section{Construction and main theorem}

Let $\N=\{0,1,2,\ldots\}$. A \emph{factor} of an infinite
word is a finite contiguous subword appearing in it. For a finite word $u$,
let $\lvert u \rvert$ be its length and $\lvert u\rvert_a$ the number of occurrences of the letter $a$.
The \emph{abelian complexity} $\rho_{\mathbf{x}}^{\rm ab}(n)$ of a
binary word $\mathbf{x}$ then counts the distinct pairs
$(\lvert u\rvert_0, \lvert u\rvert_1)$ among its length-$n$ factors \cite{Richomme&Saari&Zamboni:2011}.

\begin{lemma}
\label{lem:computing-abcompl}
Let $\mathbf{x}$ be an infinite binary word.
For all $n\ge 0$, we let $M_n$ and $L_n$, respectively, denote the maximum and minimum numbers of $1$'s in length-$n$ factors of $\mathbf{x}$.
Then, for all $n\ge 0$, we have $\rho_{\mathbf{x}}^{\rm ab}(n)=M_n-L_n+1$.
\end{lemma}
\begin{proof}
Sliding a window one position changes its count by at most $1$, so all intermediate counts occur. 
\end{proof}

\begin{remark} For binary sequences the
quantity $M_n - L_n$ is exactly the balance function of
Adamczewski \cite{Adamczewski:2003}.
\end{remark}

For $n\in\N$, define $d_3(n)$ to be the number of occurrences of the digit $3$ in the base-$4$ expansion of $n$, i.e.,
\[
\text{ if }n=\sum_{i=0}^{\ell} n_i\, 4^i \text{ with }n_0,\ldots,n_\ell\in\{0,1,2,3\}\text{, then }d_3(n)=\lvert n_\ell\cdots n_0\rvert_3.
\]
Define the sequences ${\bf s} = (s(n))_{n \geq 0}$ 
and ${\bf b} = (b(n))_{n \geq 0 }$ by
\begin{align*}
 s(n)=d_3(n)\bmod3 \quad\text{ and }\quad
 b(n) &=\begin{cases}
 1, & \text{if $s(n)=0$;}\\
 0, & \text{if $s(n)\neq 0$.}\end{cases}
\end{align*}
The first few values of these sequences can be found in Table~\ref{tab:first-values-of-sequences}. 
\begin{table}[H]
    \centering
    \[
    \begin{array}{c|c}
        \mathbf{x} & x(0)x(1)x(2)\cdots x(63) \\[1ex]
        \hline
        \mathbf{d}_3 & \renewcommand{\arraystretch}{1.5}0001000100011112000100010001111200010001000111121112111211122223 \\[1ex]
        \mathbf{s} & 0001000100011112000100010001111200010001000111121112111211122220 \\[1ex]
        \mathbf{b} & 1110111011100000111011101110000011101110111000000000000000000001  
    \end{array}
    \]
    \caption{The values of the sequences $\mathbf{d}_3$, $\mathbf{s}$, and $\mathbf{b}$ for $0\le n\le 63$.}
    \label{tab:first-values-of-sequences}
\end{table}
Note that $(d_3(n))_{n\ge 0}$ and $(s(n))_{n\ge 0}$ differ for the first time at $n=63=4^3-1$, for which we have $d_3(63)=3$ and thus $s(63)=0$.
The binary word $\mathbf{b}$ records the positions
of $0$ in $\mathbf{s}$ (also known as the characteristic sequence of $0$'s in $\mathbf{s}$).

Since $d_3(4n+j)=d_3(n)$ for $j<3$ and $d_3(4n+3)=d_3(n)+1$ for all $n\ge 0$, we have the following result.

\begin{lemma}
\label{lem:s-fixed-pt}
The sequence $\mathbf{s}$ is the fixed point starting
with $0$ of the morphism $\sigma \colon \{0,1,2\}^* \to \{0,1,2\}^*, a \mapsto aaa(a+1)$, where letters are added modulo $3$.
\end{lemma}

\begin{remark}
It may be interesting to study the morphisms $\sigma_r$, $r\geq 1$, where $a\mapsto a^r(a+1)$ and letters are taken mod $r$. When $r=1$ we obtain the Thue--Morse morphism for which the corresponding sequence is $2$-automatic and its abelian complexity is $2$-regular, as mentioned previously. When $r=3$ and $A=\{0,1,2\}$, we obtain the morphism of Lemma~\ref{lem:s-fixed-pt} for which we have Theorem~\ref{thm:main}. What can be said for other values of $r$ and $A$?
\end{remark}

With $\sigma$ one can associate the matrix $\mathcal{M}_\sigma$ in $\N^{3\times 3}$ where $(\mathcal{M}_\sigma)_{i,j}=\lvert\sigma(j)\rvert_i$ for $i,j\in\{0,1,2\}$, i.e.,
\begin{align}
\label{eq:matrix-sigma}
\mathcal{M}_\sigma=\begin{pmatrix}
    3& 0& 1\\ 1& 3 & 0\\ 0& 1& 3\\
\end{pmatrix}.
\end{align}
Its eigenvalues are $4$, $\lambda$ and $\overline{\lambda}$, where $\lambda=\frac{1}{2} \left(5+ i \sqrt{3}\right)$.

\begin{lemma}
\label{lem:pair-of-letters}
Every pair of letters in $\{0,1,2\}$ occurs in $\mathbf{s}$.
\end{lemma}
\begin{proof}
Observe that all three letters occur in $\sigma^2(0)=0001000100011112$.
Now, for $a\in\{0,1,2\}$, the pairs
$aa$ and $a(a+1)$ occur inside $\sigma(a)$, and $(a+1)a$
occurs across the boundary in $\sigma(aa)$.
\end{proof}

\begin{lemma}
\label{lem:b-2auto}
The sequence $\mathbf{b}$ is $2$-automatic.
\end{lemma}
\begin{proof}
The sequence $\mathbf{s}$ is $4$-automatic by Lemma~\ref{lem:s-fixed-pt}, as the morphism $\sigma$ is $4$-uniform.
Since $\mathbf{b}$ is obtained from $\mathbf{s}$ by applying the coding $0\mapsto 1; 1,2\mapsto 0$, the sequence $\mathbf{b}$ is also $4$-automatic. 
It is known that, for $m \geq 1$, a sequence is $k$-automatic 
if and only if it is $k^m$-automatic; see, for example,
\cite[Thm.~6.6.4]{Allouche&Shallit:2003}.
In particular, the sequence 
$\mathbf{b}$ is $2$-automatic.
\end{proof}

We now formulate our main contribution.

\begin{theorem}\label{thm:main}
The sequence $\mathbf{b}$ is $2$-automatic, while its abelian complexity $\rho_{\mathbf{b}}^{\rm ab}$ is not $2$-regular.
\end{theorem}

\section{An exact formula for a subsequence of
\texorpdfstring{$\rho_{\mathbf{b}}^{\rm ab}$}{rho b ab}}

For all $m\in\N$, write $N_m=(N_m(0),N_m(1),N_m(2))$, where
$N_m(a)=\lvert\sigma^m(0)\rvert_a$, i.e., $N_m^\top$ is the first column of the matrix $(\mathcal{M}_\sigma)^m$. Observe that $\sum_{a=0}^2 N_m(a)=\lvert\sigma^m(0)\rvert=4^m$.
Put
\begin{equation}
\label{eqn:def-g}
g(m)=\max_a N_m(a)-\min_a N_m(a).
\end{equation}
These counts satisfy
\begin{equation}\label{eq:counts}
 N_0=(1,0,0),\qquad
 N_{m+1}(a)=3N_m(a)+N_m(a-1).
\end{equation}
The first few values of the vectors $(N_m)_{m\ge 0}$ are
\[
(1, 0, 0), (3, 1, 0), (9, 6, 1), (28, 27, 9), (93, 109, 54), (333, 420, 271),
\]
so the first few values of $(g(m))_{m\ge 0}$ are $1, 3, 8, 19, 55, 149$.
We will prove the following exact formula for the sequence $(\rho_{\mathbf{b}}^{\rm ab}(4^m))_{m\ge 0}$, whose first few values are
\[
 2,\ 5,\ 14,\ 38,\ 95,\ 260,\ 707,\ 1787,\ 4805,\ 13112.
\]

\begin{proposition}
\label{pro:formula-for-complexity}
For all $m\ge 0$, set $a(m) = \rho_{\mathbf{b}}^{\rm ab}(4^m)$.
Then, for all $m\ge 0$, we have 
$$a(m+1)-a(m)=3g(m) \ \text{ and }\ a(m)=2+3\sum_{j=0}^{m-1}g(j),$$ 
where the sequence $(g(m))_{m\ge 0}$ is defined in~\eqref{eqn:def-g}.
In particular, $a(m)\equiv2\pmod3$ for all $m\ge 0$.
\end{proposition}

To prove this result, we need a lemma.
\begin{lemma}\label{lem:extrema}
To each letter $a\in\{0,1,2\}$, assign a real weight $\weight(a)$.
Define the weight of a word $z=z_1 \cdots z_\ell$ to be the sum of its letter weights, i.e., $\weight(z) = \sum_{i=1}^\ell \weight(z_i)$.
For $a\in\{0,1,2\}$ and $m\ge 0$, let $\weight_m(a)$ be the weight of $\sigma^m(a)$, and let $H_m$
be the greatest weight of a length-$4^m$ factor of $\mathbf{s}$.
Then, for all $m\ge 0$, 
\[
 H_{m+1}=H_m+3\max_a \weight_m(a).
\]
\end{lemma}
\begin{proof}
Take $m\ge 0$. By Lemma~\ref{lem:s-fixed-pt}, $\mathbf{s}$ is a fixed point of a $4$-uniform morphism, so every length-$4^{m+1}$ factor consists of a suffix
of one length-$4^m$ substitution block, three complete blocks,
and a prefix of a fifth block. The suffix and prefix together
have length $4^m$ and form a factor of $\mathbf{s}$, since every pair of letters occurs by Lemma~\ref{lem:pair-of-letters}.
Their weight is at most $H_m$, which shows that $H_{m+1}\le H_m+3\max_a \weight_m(a)$.

For equality, for each $m\ge 0$, choose a letter $c_m$ maximizing $\weight_m$, i.e., $$\weight_m(c_m) = \max_a \weight_m(a).$$
Since $\sigma^{m+1}(a)=\sigma^m(aaa(a+1))$, we have $\weight_{m+1}(a)=3\weight_m(a)+\weight_m(a+1)$. Recall that letters are taken modulo~$3$; the alphabet is $\Z/3\Z$. We can choose the
letters of the sequence $(c_m)_{m\ge 0}$ successively with $c_{m+1}\in\{c_m-1,c_m\}$. Indeed, we have 
\[
 \weight_{m+1}(c_m)-\weight_{m+1}(c_m+1)
 =3\weight_m(c_m)-2\weight_m(c_m+1)-\weight_m(c_m-1)\ge0,
\]
where we used that $c_m+2=c_m-1$. The inequality $\weight_{m+1}(c_m)\ge \weight_{m+1}(c_m+1)$ implies that $c_{m+1}$ maximizing $\weight_{m+1}$ can be chosen in $\{c_m-1,c_m\}$.

We prove inductively that $H_m$ is attained by a factor
starting in the first block of either
$\sigma^m(c_m-1)\sigma^m(c_m)$ or
$\sigma^m(c_m)\sigma^m(c_m+1)$. In other words, there exist words $p,s,x,y$ such that $sp$ is a length-$4^m$ factor of maximal weight,  with $s$ nonempty, $\sigma^m(c_m-1)=xs$ and $\sigma^m(c_m)=py$; or $\sigma^m(c_m)=xs$ and $\sigma^m(c_m+1)=py$.

For the base case $m=0$, take the single letter $c_0$, at the start of $c_0(c_0+1)$, which occurs in $\mathbf{s}$ by Lemma~\ref{lem:pair-of-letters}. This has maximum weight by definition.

Now turn to the induction step.
Write $c=c_m$, $U_a=\sigma^m(a)$.
The word
\begin{align}
\label{eq:factor-of-len-4N}
 \sigma^{m+1}(c-1)\,\sigma^{m+1}(c)
 =U_{c-1}\,U_{c-1}\,U_{c-1}\,U_c\,U_c\,U_c\,U_c\,U_{c+1}   
\end{align}
is a factor of $\mathbf{s}$ since $(c-1)c$ occurs in $\mathbf{s}$ by Lemma~\ref{lem:pair-of-letters}.
By the induction hypothesis, there is a maximizing factor $z$ of length $4^m$ starting in the first block of either $U_{c-1}U_c$ or $U_cU_{c+1}$.
We consider two cases.
For a word $w$, let $\pref_t(w)$ (resp., $\suff_t (w)$) be the length-$t$ prefix (resp., suffix) of $w$. 

\bigskip\noindent\textbf{Case 1}. Suppose that $z$ starts in the first block of $U_{c-1}U_c$.
We have $z=\suff_{4^m-t}(U_{c-1})\pref_t(U_c)$ for some $0\le t<4^m$.
In the length-$2\cdot 4^{m+1}$ factor highlighted in~\eqref{eq:factor-of-len-4N}, consider the factor $\suff_{4^m-t}\,(U_{c-1})\,U_c\,U_c\,U_c\,\pref_t(U_c)$ starting in position $2 \cdot 4^m+t$.
Since the weight function is additive, the weight of this length-$4^{m+1}$ factor is 
\begin{align*}
\weight(\suff_{4^m-t}(U_{c-1}))+\weight(U_cU_cU_c)+\weight(\pref_t(U_c))
&= \weight(z) + 3\weight(U_c)\\
&= H_m + 3 \max_a \weight_m(a),
\end{align*}
by definition of $c$ and $z$. Since we have found a length-$4^{m+1}$ factor with this weight, we conclude that $H_{m+1}\ge H_m+3\max_a \weight_m(a)$.

\bigskip\noindent\textbf{Case 2}. Suppose that $z$ starts in the first block of $U_cU_{c+1}$.
We have $z=\suff_{4^m-t} (U_c)\, \pref_t(U_{c+1})$ for some $0\le t<4^m$.
In the length-$2\cdot 4^{m+1}$ factor highlighted in~\eqref{eq:factor-of-len-4N}, consider the factor $\suff_{4^m-t}(U_c) \,U_c\, U_c \,U_c\, \pref_t(U_{c+1})$ starting in position $3 \cdot 4^m+t$.
Similarly to the first case, its weight is $H_m + 3 \max_a \weight_m(a)$.
\end{proof}

We turn to the proof of Proposition~\ref{pro:formula-for-complexity}.

\begin{proof}[Proof of Proposition~\ref{pro:formula-for-complexity}]
For all $n\ge 0$, let $M_n$ and $L_n$, respectively, denote the maximum and minimum numbers of $1$'s in length-$n$ factors of $\mathbf{b}$.
By Lemma~\ref{lem:computing-abcompl}, we have
\begin{equation}
\label{eq:first-eq-for-am}
a(m) = \rho_{\mathbf{b}}^{\rm ab}(4^m)=M_{4^m}-L_{4^m}+1    
\end{equation}
for all $m\ge 0$.
To obtain the first formula of the statement, we apply Lemma~\ref{lem:extrema} to two different sequences of weights.

First, apply Lemma~\ref{lem:extrema} to the weights $\weight(0)=1$ and $\weight(1)=\weight(2)=0$.
This particular weight simply counts the number of $0$'s in words.
With the notation of this lemma, $H_m$ then gives the maximum number of $0$'s in length-$4^m$ factors of $\mathbf{s}$, which is $M_{4^m}$ since $\mathbf{b}$ is the characteristic sequence of $0$'s in $\mathbf{s}$.
Now we examine the weight $\weight_m(a)$ of $\sigma^m(a)$.
Note that each letter $0,1,2$ in $\sigma^m(0)$ corresponds to the letter $a,a+1,a+2$ (if we count modulo $3$) in $\sigma^m(a)$.
Therefore, the counts of $0,1,2$ in $\sigma^m(a)$ are permutations of the counts of $0,1,2$ in $\sigma^m(0)$. In other words, the columns of $(\mathcal{M}_\sigma)^m$ are cyclic permutations of each other.
Consequently, recalling the definition of $N_m(a)=\lvert\sigma^m(0)\rvert_a$ at the beginning of the section, Lemma~\ref{lem:extrema} with $H_{m+1}=H_m+3\max_a \weight_m(a)$ yields the recurrence relation  
$$M_{4^{m+1}}=M_{4^m}+3\max_a N_m(a).$$

Second, applying Lemma~\ref{lem:extrema} to the weights $\weight(0)=-1$ and $\weight(1)=\weight(2)=0$ gives in a similar way 
$$L_{4^{m+1}}=L_{4^m}+3\min_a N_m(a).$$

Together with~\eqref{eq:first-eq-for-am}, these recurrence relations give
\begin{align*}
a(m+1) 
=  M_{4^{m+1}} - L_{4^{m+1}} +1
&=  M_{4^m} - L_{4^m} + 1 + 3\left(\max_a N_m(a) -\min_ a N_m(a)\right)\\
&= a(m) + 3g(m),
\end{align*}
which is the first equality of the statement.
For the second, use this recurrence repeatedly and recall that $a(0)=2$.
\end{proof}

\section{Parity and non-regularity}

In this section, we consider the complex numbers
\[
\omega=e^{2\pi i/3} = \frac{-1}{2} + i \frac{\sqrt{3}}{2} \quad \text{and} \quad
\lambda=3+\omega=(5+i\sqrt3)/2=\sqrt7e^{i\theta},
\]
where $\theta = \arctan(\sqrt{3}/5)$ is the argument of $\lambda$ (recall that $\lambda$ is an eigenvalue of $\mathcal{M}_\sigma$ defined in~\eqref{eq:matrix-sigma}).

\begin{lemma}
\label{lem:sum-and-weightedsum}
We have $\sum_{a=0}^2 N_m(a) = 4^m$ and $\sum_{a=0}^2 \omega^a N_m(a) =\lambda^m$.
\end{lemma}
\begin{proof}
The first equality follows by definition, as we have
\[
\sum_{a=0}^2 N_m(a) = \sum_{a=0}^2 \lvert\sigma^m(0)\rvert_a = \lvert \sigma^m(0) \rvert = 4^m.
\]
For the second equality, we proceed by induction on $m\ge 0$.
The case $m=0$ is clear, since $N_0=(1,0,0)$ and for $m=1$, we have $N_1=(3,1,0)$, so
\[
\sum_{a=0}^2 \omega^a N_1(a) = \omega^0 \cdot 3 + \omega \cdot 1 + \omega^2 \cdot 0 = \lambda
\]
by definition.
Now assume that the second equality holds for $m\ge 0$.
Using the recurrence given in~\eqref{eq:counts}, the fact that $\omega^{-1}=\omega^2$ and the induction hypothesis, we get
\[
\sum_{a=0}^2 \omega^a N_{m+1}(a)
= 3 \sum_{a=0}^2 \omega^a N_m(a) + \omega \sum_{a=0}^2 \omega^{a-1} N_m(a-1) 
= (3 + \omega) \cdot \lambda^m = \lambda^{m+1},
\]
as desired.
\end{proof}

Since the matrix $\mathcal{M}_\sigma$ has only simple eigenvalues, it is diagonalizable. It is therefore straightforward to express the entries of $(\mathcal{M}_\sigma)^m$ as linear combinations of the $m$th powers of these eigenvalues. In particular, a few elementary manipulations yield the vector $N_m$.

\begin{lemma}
\label{lem:cosines}
For each letter $a\in\{0,1,2\}$, we have
\[
N_m(a)=\frac{4^m+2\cdot 7^{m/2}\cos\left(m\theta- \frac{2\pi}{3}a\right)}{3}.
\]
\end{lemma}
\begin{proof}
Using the first equality in Lemma~\ref{lem:sum-and-weightedsum} and the two equalities derived from identifying the real and imaginary parts in the second equality of Lemma~\ref{lem:sum-and-weightedsum}, we obtain a system of three equations in the unknowns $N_m(0),N_m(1),N_m(2)$, whose solution is
\begin{align}
N_m(0) &= \frac{1}{3} \left( 4^m + 2 \cdot 7^{m/2} \cos(m\theta)\right), \label{nm0}\\
N_m(1) &= \frac{1}{3} \left( 4^m - 7^{m/2} \cos(m\theta) + \sqrt{3} \cdot 7^{m/2} \sin(m\theta) \right), \label{nm1}\\
N_m(2) &= \frac{1}{3} \left( 4^m - 7^{m/2} \cos(m\theta) - \sqrt{3} \cdot  7^{m/2} \sin(m\theta) \right). \label{nm2}
\end{align}
Recalling the formula $\cos(x-y)=\cos(x)\cos(y)+\sin(x)\sin(y)$ and the particular values $\cos(\frac{2\pi}{3})=\frac{-1}{2}$, $\sin(\frac{2\pi}{3})=\frac{\sqrt{3}}{2}$, this gives the formula in the statement.
\end{proof}

\begin{lemma}
\label{lem:not-a-rational}
The real number $\frac{\theta}{\pi}$ is irrational.
\end{lemma}
\begin{proof}
Using the values $N_2(0)=9$ (see the first few values of $(N_m)_{m\ge 0}$ given after~\eqref{eq:counts}) and the equality in Lemma~\ref{lem:cosines} for $a=0$ and $m=2$, we get $\cos(2\theta) = \frac{11}{14} \in \Q$.
Niven's theorem (see \cite[Thm.~1, p.~132]{PaolilloVincenzi} 
and the discussion there) restricts rational cosines of rational 
multiples of $\pi$ to $\{0,\pm\frac{1}{2},\pm 1\}$.
This implies that $2\theta$ is not a rational multiple of $\pi$, 
which is enough.
\end{proof}

The next lemma determines the parity of the sequence $g(3n)$ in terms of $\lvert\cos(3n\theta)\rvert$.
The first few values of the latter sequence are shown in Figure~\ref{fig:cos3nt}.

\begin{lemma}
\label{lem:inequality-cosines}
For all $n\ge 1$, we have $g(3n)\equiv0\pmod2$ if and only if $\lvert \cos(3n\theta) \rvert<1/2$.
\end{lemma}
\begin{proof}
Taking the recurrence~\eqref{eq:counts} modulo $2$ shows that the sequence $(N_m)_{m\ge 1}$ cycles through the three vectors $(1,1,0),(1,0,1),(0,1,1)$ (note that $N_0=(1,0,0)$).
Thus, for $n\ge1$, $N_{3n}(0)$ is even and the other two counts $N_{3n}(1),N_{3n}(2)$ are odd.
From~\eqref{eqn:def-g}, their range $g(3n)$  is even exactly when $N_{3n}(0)$ lies
strictly between the other two quantities $N_{3n}(1),N_{3n}(2)$.
Indeed, the only orderings of the natural numbers $N_{3n}(0),N_{3n}(1),N_{3n}(2)$ giving an even range are when $N_{3n}(1) < N_{3n}(0) < N_{3n}(2)$ or $N_{3n}(2) < N_{3n}(0) < N_{3n}(1)$.
Lemma~\ref{lem:cosines} gives
\begin{align*}
 \left(N_m(0)-N_m(1)\right)\left(N_m(0)-N_m(2)\right)
&= \frac{7^m}{9} \left( 3 \cos(m\theta) - \sqrt{3} \sin(m\theta)\right) \left( 3 \cos(m\theta) + \sqrt{3}  \sin(m\theta) \right) \\
& =\frac{7^m}{3}\bigl(3\cos^2(m\theta)-\sin^2(m\theta)\bigr) \\
& =\frac{7^m}{3}\bigl(4\cos^2(m\theta)-1\bigr).
\end{align*}
Consequently,
\begin{equation}\label{eq:parity}
 g(3n)\equiv0\pmod2
 \quad\Longleftrightarrow\quad \frac{7^{3n}}{3}\bigl(4\cos^2(3n\theta)-1\bigr) < 0
 \quad\Longleftrightarrow\quad \lvert \cos(3n\theta) \rvert <1/2,
\end{equation}
as desired.
\end{proof}

\begin{figure}[H]
    \centering
    \includegraphics[width=0.5\linewidth]{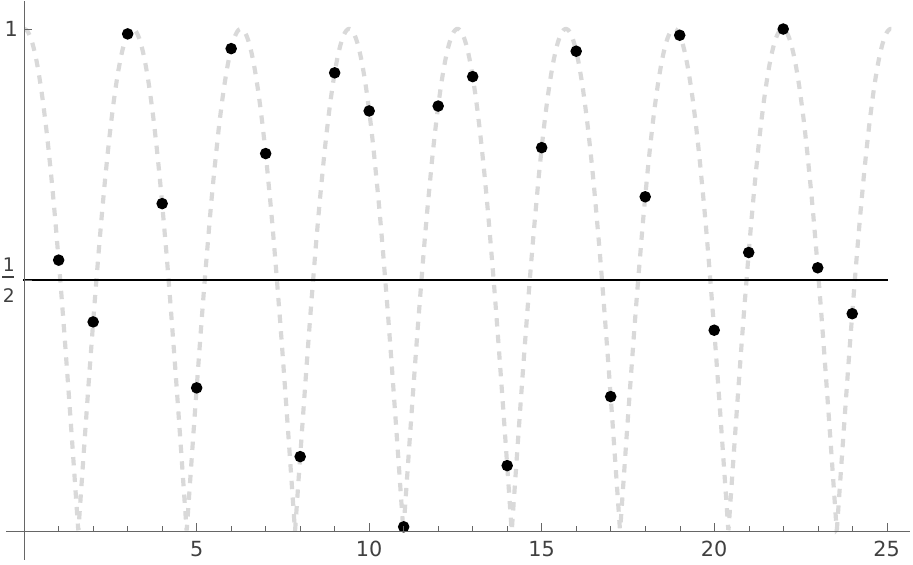}
    \caption{The first few values of $(\lvert \cos(3n\theta) \rvert)_{n\ge 1}$.}
    \label{fig:cos3nt}
\end{figure}

Recall that a sequence $(x(n))_{n\ge 0}$ is \emph{eventually periodic} if $x(n+p)=x(n)$ for some $p\ge1$ and all sufficiently large $n$.

\begin{proposition}
\label{pro:eventually-periodic}
The sequence $(\rho_{\mathbf{b}}^{\rm ab}(4^m)\bmod2)_{m\ge 0}$ is not eventually periodic.
\end{proposition}
\begin{proof}
For the sake of conciseness, set $a(m)=\rho_{\mathbf{b}}^{\rm ab}(4^m)$ for all $m\ge 0$.

Recall that Lemma~\ref{lem:not-a-rational} implies that $\theta$ is not a rational multiple of $\pi$.
Fix a positive integer $p$.
The angles $(3pn\theta)_{n\ge 1}$ are
dense modulo $\pi$, even after deleting finitely many terms.
This is the density theorem for irrational rotations 
(Kronecker's theorem); 
see, e.g., \cite[Thm.~439, p.~376]{HardyWright}. 
The conditions $\lvert \cos t \rvert<1/2$ and 
$\lvert\cos t \rvert >1/2$ each hold on nonempty open intervals. 
Lemma~\ref{lem:inequality-cosines} implies that $(g(3pn)\bmod2)_{n\ge1}$
takes both values infinitely often. It follows that
$(g(m)\bmod2)_{m\ge 0}$ cannot be eventually periodic with any 
period $p$. By Proposition~\ref{pro:formula-for-complexity}, the same is 
true of $(a(m)\bmod2)_{m\ge 0}$, since eventual periodicity is 
preserved by taking first differences. This proves the statement.
\end{proof}

We have everything in our hands to prove the main theorem.

\begin{proof}[Proof of Theorem~\ref{thm:main}]
The fact that the sequence $\mathbf{b}$ is $2$-automatic follows from Lemma~\ref{lem:b-2auto}.
Now assume, to the contrary, that 
$\rho_{\mathbf{b}}^{\rm ab}$ is $2$-regular.
By~\cite[Cor.~16.1.6]{Allouche&Shallit:2003}, the sequence $(\rho_{\mathbf{b}}^{\rm ab}(n) \bmod{2})_{n\ge 0}$ is $2$-automatic.

Observe that, if a sequence $(x(n))_{n\ge 0}$ is $2$-automatic, then the sequence $(x(4^n))_{n\ge 0}$ is eventually periodic. Indeed, there is a DFAO~$\mathcal{A}$ that, when fed with  the binary expansion of $n$, outputs $x(n)$. Since the binary expansion of $4^n$ is $1(00)^n$ for every $n\ge 0$, the sequence of states $(q_0\cdot 1(00)^n)_{n\ge 0}$ reached from the initial state of~$\mathcal{A}$ is eventually periodic. As the state set is finite, this sequence is eventually periodic, and hence so is the corresponding output sequence $(x(4^n))_{n\ge 0}$.

Applying this to our case, the sequence $(\rho_{\mathbf{b}}^{\rm ab}(4^n) \bmod{2})_{n\ge 0}$ is eventually periodic, contradicting Proposition~\ref{pro:eventually-periodic}.
\end{proof}

\section*{Acknowledgments}
Narad Rampersad is supported by NSERC grant RGPIN-2025-04076.
Michel Rigo is supported by the FNRS grant T.196.23 (PDR). 
Jeffrey Shallit is supported by NSERC grant RGPIN-2024-03725.
Manon Stipulanti is an FNRS Research Associate (CQ) supported by the Research grant 1.C.104.24F.

\section*{Declaration of AI Usage}

Part of this paper was obtained through extended discussion with GPT Astra 6, a large language AI model.
We have extensively rewritten the original arguments.  The authors take responsibility for all of the contents.

\bibliographystyle{alpha}
\bibliography{biblio}

\end{document}